\documentclass{llncs}
\usepackage{epsfig}
\usepackage{amsmath}
\usepackage{amssymb}
\usepackage{algorithmic}
\usepackage{bigstrut}
\usepackage{fullpage}

\def\idtt#1{\ensuremath{\mathtt{#1}}}

\def\rankop{\idtt{rank}}

\def\accessop{\idtt{access}}
\def\polylog{\idtt{polylog}}

\def\insertop{\idtt{insert}}
\def\deleteop{\idtt{delete}}

\def\rangecount{\idtt{range\_count}}

\begin{document}

\title{Dynamic Range Selection in Linear Space\thanks{This work was
    supported by NSERC and the Canada Research Chairs Program.}}

\author{Meng He \and J. Ian Munro \and Patrick K. Nicholson}

\institute{Faculty of Computer Science, Dalhousie University, Canada \and David R. Cheriton School of Computer Science, University of Waterloo, Canada, \\
 \email mhe@cs.dal.ca, \{imunro, p3nichol\}@uwaterloo.ca
}

\maketitle
\begin{abstract}
  Given a set $S$ of $n$ points in the plane, we consider the problem
  of answering range selection queries on $S$: that is, given an
  arbitrary $x$-range $Q$ and an integer $k > 0$, return the $k$-th
  smallest $y$-coordinate from the set of points that have
  $x$-coordinates in $Q$. We present a linear space data structure
  that maintains a dynamic set of $n$ points in the plane with real
  coordinates, and supports range selection queries in $O((\lg n / \lg
  \lg n)^2)$ time, as well as insertions and deletions in $O((\lg n /
  \lg \lg n)^2)$ amortized time.  The space usage of this data
  structure is an $\Theta(\lg n / \lg \lg n)$ factor improvement over
  the previous best result, while maintaining asymptotically matching
  query and update times.  We also present a succinct data structure
  that supports range selection queries on a dynamic array of $n$
  values drawn from a bounded universe.
\end{abstract}

\section{Introduction}

The problem of finding the \emph{median} value in a data set is a
staple problem in computer science, and is given a thorough treatment
in modern textbooks~\cite{CSRL01}.  In this paper we study a dynamic
data structure variant of this problem in which we are given a set $S$
of $n$ points in the plane.  The \emph{dynamic range median problem}
is to construct a data structure to represent $S$ such that we can
support \emph{range median queries}: that is, given an arbitrary range $Q
= [x_1,x_2]$, return the median $y$-coordinate from the set of points
that have $x$-coordinates in $Q$.  Furthermore, the data structure
must support insertions of points into, as well as deletions from, the
set $S$.  We may also generalize our data structure to support
\emph{range selection queries}: that is, given an arbitrary $x$-range $Q
= [x_1,x_2]$ and an integer $k > 0$, return the $k$-th smallest
$y$-coordinate from the set of points that have $x$-coordinates in
$Q$.

In addition to being a challenging theoretical problem, the range
median and selection problems have several practical applications in
the areas of image processing~\cite{GW93}, Internet advertising,
network traffic analysis, and measuring real-estate prices in a
region~\cite{HM08}.

In previous work, the data structures designed for the above problems
that support queries and updates in polylogarithmic time require
superlinear space~\cite{BGJS10}.  In this paper, we focus on designing
linear space dynamic range selection data structures, without
sacrificing query or update time.  We also consider the problem of
designing \emph{succinct data structures} that support range selection
queries on a dynamic array of values, drawn from a bounded universe:
here ``succinct'' means that the space occupied by our data structure
is close to the information-theoretic lower bound of representing the
array of values~\cite{J89}.

\subsection{Previous Work\label{sec:previousresults}}


\paragraph{Static Case:}

The static range median and selection problems have been studied by
many different authors in recent
years~\cite{BKMT05,KMS05,HM08,P08,PG09,GPT09,GS09,BJ09,BGJS10,JL11}.
In these problems we consider the $n$ points to be in an array: that
is, the points have $x$-coordinates $\{1, ..., n\}$.  We now summarize
the upper and lower bounds for the static problem.  In the remainder
of this paper we assume the word-RAM model of computation with word
size $w = \Omega(\lg n)$ bits.

For exact range medians in constant time, there have been several
iterations of near-quadratic space data
structures~\cite{KMS05,P08,PG09}. For linear space data structures,
Gfeller and Sanders~\cite{GS09} showed that range median queries could
be supported in $O(\lg n)$ time\footnote{In this paper we use $\lg n$
  to denote $\log_2 n$.}, and Gagie et al.~\cite{GPT09} showed that
selection queries could be supported in $O(\lg \sigma)$ time using a
wavelet tree, where $\sigma$ is the number of distinct $y$-coordinates
in the set of points.  Optimal upper bounds of $O(\lg n / \lg \lg n)$
time for range median queries have since been achieved by Brodal et
al.~\cite{BJ09,BGJS10}, and lower bounds by J{\o}rgensen and
Larsen~\cite{JL11}; the latter proved a cell-probe lower bound of
$\Omega(\lg n/ \lg \lg n)$ time for any static range selection data
structure using $O(n \lg^{O(1)} n)$ bits of space.  In the case of
range selection when $k$ is fixed for all queries, J{\o}rgensen and
Larsen proved a cell-probe lower bound of $\Omega(\lg k / \lg \lg n)$
time for any data structure using $O(n \lg^{O(1)}n)$
space~\cite{JL11}.  Furthermore, they presented an adaptive data
structure for range selection, where $k$ is given at query time, that
matches their lower bound, except when $k=2^{o(\lg^2\lg
  n)}$~\cite{JL11}.  Finally, Bose et al.~\cite{BKMT05} studied the
problem of finding \emph{approximate range medians}.  A
$c$-approximate median of range $[i..j]$ is a value of rank between
$\frac{1}{c} \times \lceil \frac{j - i + 1}{2} \rceil$ and $(2 -
\frac{1}{c}) \times \lceil \frac{j - i + 1}{2} \rceil$, for $c > 1$.

\paragraph{Dynamic Case:}

Gfeller and Sanders~\cite{GS09} presented an $O(n \lg n)$ space data
structure for the range median problem that supports queries in
$O(\lg^2 n)$ time and insertions and deletions in $O(\lg^2n)$
amortized time. Later, Brodal et al.~\cite{BJ09,BGJS10} presented an
$O(n\lg n / \lg \lg n)$ space data structure for the dynamic range
selection problem that answers range queries in $O((\lg n/\lg \lg
n)^2)$ time and insertion and deletions in $O((\lg n/\lg \lg n)^2)$
amortized time.  They also show a reduction from the \emph{marked
  ancestor problem}~\cite{AHR98} to the dynamic range median problem.
This reduction shows that $\Omega(\lg n/\lg \lg n)$ query time is
required for any data structure with polylogarithmic update time.
Thus, there is still a gap of $\Theta(\lg n / \lg \lg n)$ time between
the upper and lower bounds for linear and near linear space data
structures. 

In the restricted case where the input is a dynamic array $A$ of $n$
values drawn from a bounded universe, $[1,\sigma]$, it is possible to
answer range selection queries using a dynamic wavelet tree, such as
the succinct dynamic string data structure of He and
Munro~\cite{HM10}.  This data structure uses $n H_0(A) + o(n \lg
\sigma) + O(w)$ bits\footnote{$H_0(A)$ denotes the 0th-order empirical
  entropy of the multiset of values stored in $A$.  Note that $H_0(A)
  \le \lg \sigma$ always holds.} of space, the query time is
$O(\frac{\lg n \lg \sigma}{\lg \lg n})$, and the update time is
$O(\frac{\lg n}{\lg \lg n}(\frac{\lg \sigma}{\lg \lg n} + 1))$.

\subsection{Our Results\label{sec:ourresults}}

In Section~\ref{sec:lineards}, we present a \emph{linear space} data
structure for the dynamic range selection problem that answers queries
in $O((\lg n /\lg \lg n)^2)$ time, and performs updates in $O((\lg
n/\lg \lg n)^2)$ amortized time. This data structure can be used to
represent point sets in which the points have \emph{real coordinates}.
In other words, we only assume that the coordinates of the points can
be compared in constant time.  This improves the space usage of the
previous best data structure by a factor of $\Theta(\lg n / \lg \lg
n)$~\cite{BGJS10}, while maintaining query and update time.

In Section~\ref{sec:dynarray}, we present a succinct data structure
that supports range selection queries on a dynamic array $A$ of values
drawn from a bounded universe, $[1..\sigma]$.  The data structure
occupies $n H_0(A) + o(n \lg \sigma) + O(w)$ bits, and supports
queries in $O(\frac{\lg n}{\lg \lg n} (\frac{\lg \sigma}{\lg \lg
  \sigma }))$ time, and insertions and deletions in $O(\frac{\lg
  n}{\lg \lg n} (\frac{\lg \sigma}{\lg \lg \sigma}))$ amortized time.
This is a $\Theta(\lg \lg \sigma)$ improvement\footnote{The
  preliminary version of this paper that appeared in ISAAC 2011
  erroneously stated the bound for our succinct data structure as
  being an $\Theta(\lg \lg n)$ improvement.  We thank Gelin Zhou for
  pointing out this error.} in query time over the dynamic wavelet
tree, and thus closes the space gap between the dynamic wavelet tree
solution and that of Brodal et al.~\cite{BGJS10}.

\section{Linear Space Data Structure\label{sec:lineards}}

In this section we describe a linear space data structure for the
dynamic range selection problem.  Our data structure follows the same
general approach as the dynamic data structure of Brodal et
al.~\cite{BGJS10}.  However, we make several important changes, and
use several other auxiliary data structures, in order to improve the
space by a factor of $\Theta(\lg n / \lg \lg n)$.

The main data structure is a weight balanced B-tree~\cite{AV03}, $T$,
with branching parameter $\Theta(\lg^{\varepsilon_1}n)$, for $0 <
\varepsilon_1 < 1/2$, and leaf parameter $1$.  The tree $T$ stores the
points in $S$ at its leaves, sorted in non-decreasing order of
$y$-coordinate\footnote{Throughout this paper, whenever we order a
  list based on $y$-coordinate, it is assumed that we break ties using
  the $x$-coordinate, and vice versa.}.  The height of $T$ is $h_1 =
\Theta(\lg n /\lg \lg n)$ \emph{levels}, and we assign numbers to the
levels starting with level $1$ which contains the root node, down to
level $h_1$ which contains the leaves of $T$.  Inside each internal
node $v \in T$, we store the smallest and largest $y$-coordinates in
$T(v)$.  Using these values we can acquire the path from the root of
$T$ to the leaf representing an arbitrary point contained in $S$ in
$O(\lg n)$ time; a binary search over the values stored in the
children of an arbitrary internal node requires $O(\lg \lg n)$ time
per level.

Following Brodal et al.~\cite{BGJS10}, we store a \emph{ranking tree}
$R(v)$ inside each internal node $v \in T$.  The purpose of the
ranking tree $R(v)$ is to allow us to efficiently make a branching
decision in the main tree $T$, at node $v$.  Let $T(v)$ denote the
subtree rooted at node $v$.  The ranking tree $R(v)$ represents all of
the points stored in the leaves of $T(v)$, sorted in non-decreasing
order of $x$-coordinate.  The fundamental difference between our
ranking trees, and those of Brodal et al.~\cite{BGJS10}, is that ours
are more space efficient.  Specifically, in order to achieve linear
space, we must ensure that the ranking trees stored in each level of
$T$ occupy no more than $O(n \lg \lg n)$ bits in total, since there
are $O(\lg n / \lg \lg n)$ levels in $T$.  We describe the ranking
trees in detail in Section~\ref{sec:serank}, but first discuss some
auxiliary data structures we require in addition to $T$.

We construct a red-black tree $S_x$ that stores the points in $S$ at
its leaves, sorted in non-decreasing order of $x$-coordinate.  As
in~\cite{BGJS10}, we augment the red-black tree $S_x$ to store, in
each node $v$, the count of how many points are stored in $T(v_1)$ and
$T(v_2)$, where $v_1$ and $v_2$ are the two children of $v$.  Using
these counts, $S_x$ can be used to map any query $[x_1,x_2]$ into
$r_1$, the rank of the successor of $x_1$ in $S$, and $r_2$, the rank
of the predecessor of $x_2$ in $S$.  These ranking queries, as well as
insertions and deletions into $S_x$, take $O(\lg n)$ time.

\begin{figure}
\centering
\includegraphics[scale=1.2]{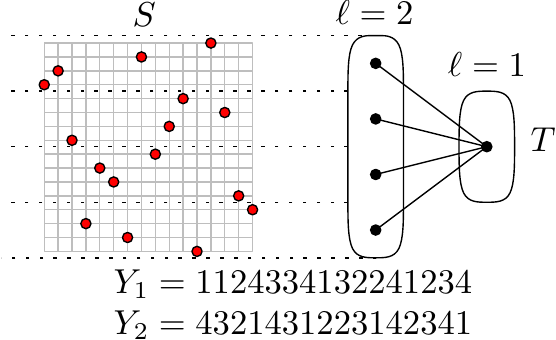}
\caption{\label{fig:stringfig}The top two levels of an example tree
  $T$, and the corresponding strings $Y_1$ and $Y_2$ for these
  levels. Each node at level $2$ has exactly $4$ children.}
\end{figure}

We also store a string $Y(v)$ for each node $v$ in $T$.  This string
consists of all of the of the points in $T(v)$ sorted in
non-decreasing order of $x$-coordinate, where each point is
represented by the index of the child of node $v$'s subtree in which
they are contained, i.e., an integer bounded by
$O(\lg^{\varepsilon_1}n)$.  However, for technical reasons, instead of
storing each string with each node $v \in T$, we concatenate all the
strings $Y(v)$ for each node $v$ at level $\ell$ in $T$ into a string
of length $n$, denoted by $Y_\ell$.  Each chunk of string $Y_\ell$
from left to right represents some node $v$ in level $\ell$ of $T$
from left to right within the level.  See Figure~\ref{fig:stringfig}
for an illustration.  We represent each string $Y_\ell$ using the
dynamic succinct data structure for representing strings of He and
Munro~\cite{HM11}.  Depending on the context, we refer to both the
string, and also the data structure that represents the string, as
$Y_\ell$.  Consider the following operations on the string $Y_\ell$:

\begin{itemize}
\item $\accessop(Y_\ell, i)$, which returns the $i$-th integer,
  $Y_\ell[i]$, in $Y_\ell$;
\item $\rankop_{\alpha}(Y_\ell, i)$, which returns the number of
  occurrences of integer $\alpha$ in $Y_\ell[1..i]$;
\item $\rangecount(Y_\ell, x_1, x_2, y_1, y_2)$, which returns the
  total number of entries in $Y_\ell[x_1..x_2]$ whose values are in
  the range $[y_1..y_2]$;
\item $\insertop_{\alpha}(Y_\ell, i)$, which inserts integer $\alpha$
  between $Y_\ell[i-1]$ and $Y_\ell[i]$;
\item $\deleteop(Y_\ell, i)$, which deletes $Y_\ell[i]$ from $Y_\ell$. 
\end{itemize}

Let $W = \lceil \frac{\lceil \lg n \rceil^2}{\lg \lceil \lg
  n\rceil}\rceil$.  The following lemma summarized the functionality
of these data structures for succinct \emph{dynamic strings} over
small universe:

\begin{lemma}[\cite{HM11}]
\label{lem:smallalphabet}
Under the word RAM model with word size $w = \Omega(\lg n)$, a string
$Y_\ell[1..n]$ of values from a bounded universe $[1..\sigma]$, where
$\sigma = O({\lg^{\mu} n})$ for any constant $\mu \in (0,1)$, can be
represented using $nH_0(Y_\ell) + O(\frac{n\lg\sigma\lg\lg
  n}{\sqrt{\lg n}}) + O(w)$ bits to support $\accessop$, $\rankop$,
$\rangecount$, $\insertop$ and $\deleteop$ in $O(\frac{\lg n}{\lg\lg
  n})$ time.  Furthermore, we can perform a batch of $m$ update
operations in $O(m)$ time on a substring $Y_\ell[i..i+m-1]$ in which
the $j$-th update operation changes the value of $Y_\ell[i+j - 1]$,
provided that $m > \frac{5W}{\lg \sigma}$.
\end{lemma}

The data structure summarized by the previous lemma is, roughly, a
B-tree constructed over the string $Y_\ell[1..n]$, in which each leaf
stores a \emph{superblock}, which is a substring of $Y_\ell[1..n]$ of
length at most $2W$ bits.  We mention this because the ranking tree
stored in each node of $T$ will implicitly reference to these
superblocks instead of storing leaves.  Thus, the leaves of the
dynamic string at level $\ell$ are \emph{shared} with the ranking
trees stored in nodes at level $\ell$.

As for their functionality, these dynamic string data structures
$Y_\ell$ are used to translate the ranks $r_1$ and $r_2$ into ranks
within a restricted subset of the points when we navigate a path from
the root of $T$ to a leaf.  The space occupied by these strings is
$O((n \lg (\lg^{\varepsilon_1} n) +w) \times \lg n / \lg \lg n)$ bits,
which is $O(n)$ words.  We present the following lemma:

\begin{lemma}
\label{lem:space-1}
Ignoring the ranking trees stored in each node of $T$, the data
structures described in this section occupy $O(n)$ words.
\end{lemma}

In the next section we discuss the technical details of our
space-efficient ranking tree.  The key idea to avoid using linear
space per ranking tree is to \emph{not} actually store the points in
the leaves of the ranking tree, sorted in non-decreasing order of
$x$-coordinate.  Instead, for each point $p$ in ranking tree $R(v)$,
we implicitly reference the the string $Y(v)$, which stores the index
of the child of $v$ that contains $p$.

\subsection{Space Efficient Ranking Trees\label{sec:serank}}

Each ranking tree $R(v)$ is a weight balanced B-tree with branching
parameter $\lg^{\varepsilon_2}n$, where $0 < \varepsilon_2 < 1 -
\varepsilon_1$, and leaf parameter $\Theta(W/\lg\lceil \lg n\rceil) =
\Theta((\lg n / \lg \lg n)^2)$.  Thus, $R(v)$ has height $\Theta(\lg n
/ \lg \lg n)$, and each leaf implicitly represents a substring of
$Y(v)$, which is actually stored in one of the dynamic strings,
$Y_\ell$.

\paragraph{Internal Nodes:}
Inside each internal node $u$ in $R(v)$, let $q_i$ denote the number
of points stored in the subtree rooted at the $i$-th child of $u$, for
$1 \le i \le f_2$, where $f_2$ is the degree of $u$.  We store a
\emph{searchable partial sums structure}~\cite{RRR02} for the sequence
$Q = \{q_1, ..., q_{f_2} \}$.  This data structure will allow us to
efficiently navigate from the root of $R(v)$ to the leaf containing
the point of $x$-coordinate rank $r$.  The functionality of this data
structure is summarized in the following lemma:

\begin{lemma}[\cite{RRR02}]\label{lem:RRR}
Suppose the word size is $\Omega(\lg n)$.  A sequence $Q$ of
$O(\lg^{\mu} n)$ nonnegative integers of $O(\lg n)$ bits each, for any
constant $\mu \in (0,1)$, can be represented in $O(\lg^{1+\mu}n)$ bits
and support the following operations in $O(1)$ time:
\begin{itemize}
\item $\idtt{sum}(Q,i)$ which returns $\sum_{j=1}^{i}Q[j]$,
\item $\idtt{search}(Q,x)$ which returns the smallest $i$ such that
  $\idtt{sum}(Q,i) \ge x$,
\item $\idtt{modify}(Q,i,\delta)$ which sets $Q[i]$ to
  $Q[i]+\delta$, where $|\delta| \le \lg n$.
\end{itemize}
This data structure can be constructed in $O(\lg^{\mu}n)$ time, and it
requires a precomputed universal table of size $O(n^{\mu'})$ bits for
any fixed $\mu' > 0$.
\end{lemma}

We also store the \emph{matrix structure} of Brodal et
al.~\cite{BGJS10} in each internal each node $u$ of the ranking tree.
Let $f_1 = \Theta(\lg^{\varepsilon_1}n)$ denote the out-degree of node
$v \in T$, and let $T(v_1), ..., T(v_{f_1})$ denote the subtrees
rooted at the children of $v$ from left to right.  Similarly, recall
that $f_2 = \Theta(\lg^{\varepsilon_2}n)$ denotes the out-degree of $u
\in R(v)$, and let $T'(u_1), ..., T'(u_{f_2})$ be the subtrees rooted
at each child of $u$ from left to right.  These matrix structures are
a kind of partial sums data structure defined as follows; we use
roughly the same notation as~\cite{BGJS10}:

\begin{definition}[Summarizes~\cite{BGJS10}]
  A \emph{matrix structure} $M^u$ is an $f_1 \times f_2$~matrix, where
  entry $M^u_{p,q}$ stores the number of points from $\cup_{i=1}^{q}
  T'(u_i)$ that are contained in $\cup_{i = 1}^{p} T(v_i)$.  The
  matrix structure $M^u$ is stored in two ways.  The first
  representation is a standard table, where each entry is stored in
  $O(\lg n)$ bits.  In the second representation, we divide each
  column into \emph{sections} of $\Theta(\lg^{1-\varepsilon_1}n)$ bits
  --- leaving $\Theta(\lg \lg n)$ bits of overlap between the sections
  for technical reasons --- and we number the sections $s_1, ...,
  s_g$, where $g = \Theta(\lg^{\varepsilon_1} n)$.  In the second
  representation, for each column $q$, there is a \emph{packed word}
  $w^u_{q,i}$, storing section $s_i$ of each entry in column $q$.
  Again, for technical reasons, the most significant bit of each
  section stored in the packed word $w^u_{q,i}$ is padded with a zero
  bit.
\end{definition}

We defer the description of how the matrix structures are used to
guide queries until Section~\ref{sec:answering}.  For now, we just
treat these structures as a black box and summarize their properties
with the following lemma:

\begin{lemma}[\cite{BGJS10}]\label{lem:matrix-struct}
  The matrix structure $M^u$ for node $u$ in the ranking tree $R(v)$
  occupies $O(\lg^{1+\varepsilon_1+\varepsilon_2} n)$ bits, and can be
  constructed in $O(\lg^{\varepsilon_1+\varepsilon_2}n)$ time.
  Furthermore, consider an update path that goes through node $u$ when
  we insert a value into or delete a value from $R(v)$.  The matrix
  structures in each node along an update path can be updated in
  $O(1)$ amortized time per node.
\end{lemma}

\paragraph{Shared Leaves:}

Now that we have described the internal nodes of the ranking tree, we
describe the how the leaves are shared between $R(v)$ and the dynamic
string over $Y_\ell$.  To be absolutely explicit, we do not actually
store the leaves of $R(v)$: they are only conceptual. We present the
following lemma:

\begin{lemma}
\label{lem:leaf-query}
Let $u$ be a leaf in $R(v)$ and $S$ be the substring of $Y(v)$ that
$u$ represents, where each value in $S$ is in the range $[1..\sigma]$,
and $\sigma = \Theta(\lg^{\varepsilon_1}n)$.  Using a universal table
of size $O(\sqrt{n}\times \polylog(n))$ bits, for any $z \in
[1..|S|]$, an array $C_z = \{c_1, ..., c_\sigma \}$ can be computed in
$O(\lg n / \lg \lg n)$ time, where $c_i = \rankop_i(S,z)$, for $1 \le
i \le \sigma$.
\end{lemma}

\begin{proof}
Instead of explicitly storing the leaves of $R(v)$, we use the partial
sums structures along the path from the root of $R(v)$ to the parent
of leaf $u$ to produce two ranks, $r'_1$ and $r'_2$, which are the
starting and ending ranks of the substring $S$ represented by $u$ in
$Y_\ell$.  Based on the leaf parameter of $R(v)$, and the properties
of weight balanced B-trees~\cite{AV03}, $S$ can have length
$\Theta(W/\lg\lg n)$.  Following the analysis of how superblocks are
laid out in the dynamic string $Y_\ell$, this means that $S$ is stored
in a constant number of consecutive superblocks~\cite[Section
  4]{HM10a}.

Given that we know $r'_1$ and $r'_2$, we can acquire a pointer to the
first superblock that stores part of $S$ in $O(\lg n / \lg \lg n)$
time~\cite[Lemma 7]{HM10a}.  Inside each superblock the substring is
further decomposed into a list of \emph{blocks} of length $O((\lg
n)^{3/2})$ bits each, in which only the final block has free space.
In order to produce the array $C_z$, we scan the list of blocks in the
superblock up to position $z$, reading $(\lg n)/2$ bits at a time.
Since each value in $[1..\sigma]$ occupies $\varepsilon_1 \lg \lg n$
bits, we can read $\Theta(\lg n / \lg \lg n)$ values at a time.  As we
read these values, we keep a running total of the ranks of each value
in $[1..\sigma]$ up to our current position.  Let \emph{field} $b_i$
denote $\rankop_i(S,p)$, where $p$ is our current position within $S$.
Clearly, $b_i$ occupies $O(\lg \lg n)$ bits.  Furthermore, let $b =
b_1 ... b_\sigma$ be the concatenation of these fields.  Thus, the
running total, $b$, contains at most $2\lg^{\varepsilon_1} \times
\varepsilon_1 \lg \lg n$ bits, and can be stored in a constant number
of words.

In order to efficiently update our running total $b$ after reading a
$(\lg n)/2$ bits from the current block, we perform a lookup in a
universal table $A$.  Let $a$ denote the $a$-th lexicographically
smallest string of length $(\lg n)/2$ bits, over the bounded universe
$[1,\sigma]$.  Also, let $f_i(a)$ denote the frequency of symbol $i
\in [1,\sigma]$ in string $a$. In each entry of the table $A[a]$, we
store the value $b' = f_1(a)...f_\sigma(a)$.  Since both $b$ and $b'$
fit in a constant number of words, we can exploit word-level
parallelism to update the running total by summing all the fields in
$b$ and $b'$ in $O(1)$ time.

The table $A$ occupies $O(2^{(\lg n)/2} \times \lg^{\varepsilon_1} n
\times \lg \lg n) = O(\sqrt{n} \times \polylog(n))$ bits, and allows
us to process $\Theta(\lg n/\lg \lg n)$ values in $O(1)$ time.  Recall
that the entire superblock contains $\Theta((\lg n / \lg \lg n)^2)$
values.  Thus, we can return $C_z$ in $O((\lg n / \lg \lg n)^2 \times
(\lg \lg n / \lg n))$ time.
\end{proof}

We now present the following lemma regarding the space and
construction time of the ranking trees:

\begin{lemma}
\label{lem:space-2}
Each ranking tree $R(v)$ occupies $O\left( \frac{m(\lg \lg
  n)^2}{\lg^{1 - \varepsilon_1} n} + w \right)$ bits of space if
$|T(v)| = m$, and requires $O(m)$ time to construct, assuming that we
have access to the string $Y(v)$.
\end{lemma}

\begin{proof}
The space occupied by the internal nodes is $O(m (\lg \lg n)^2 /
\lg^{1-\varepsilon_1}n)$ bits, since each internal node occupies
$O(\lg^{1+\varepsilon_1+\varepsilon_2} n)$ bits by
Lemmas~\ref{lem:RRR} and~\ref{lem:matrix-struct}, and the number of
internal nodes is $O(m/\lg^{\varepsilon_2}n \times (\lg \lg n / \lg
n)^2)$.  In order to reduce the cost of the pointers between the
internal nodes in $R(v)$ to $O(\lg n)$ bits per pointer, we make use
of well known memory blocking techniques for dynamic data structures
(e.g., see~\cite[Appendix J]{HM10a}).  The main idea is to allocate a
\emph{fixed memory area} for the entire ranking tree, and perform all
updates to the ranking tree using memory from this area.  After
$\Theta(m)$ updates, we allocate a new area and copy over the entire
ranking tree.  The cost of using this memory blocking will amount to
$O(1)$ amortized time per update. Thus the overall space is $O(m (\lg
\lg n)^2 / \lg^{1-\varepsilon_1}n + w)$ bits, since we must count the
pointer to the root of the ranking tree, stored at the start of the
fixed memory area.

Now we analyze the overall construction time.  Based on
Lemmas~\ref{lem:RRR} and~\ref{lem:matrix-struct}, we can construct all
the internal nodes in $O(m)$ time, since the number of internal nodes
is $O(m/\lg^{\varepsilon_2} \times (\lg \lg n / \lg n)^2)$, and each
requires $o(\lg^{1+\varepsilon_1 + \varepsilon_2}n) +
O(\lg^{\varepsilon_2} n)$ time to construct. 
\end{proof}

\begin{remark}
Note that the discussion in this section implies that we need not
store ranking trees for nodes $v\in T$, where $|T(v)| = O(\lg n / \lg
\lg n)^2$.  Instead, we can directly query the dynamic string $Y_\ell$
using Lemma~\ref{lem:leaf-query} in $O(\lg n / \lg \lg n)$ time to
make a branching decision in $T$.  This will be important in
Section~\ref{sec:dynarray}, since it significantly reduces the number
of pointers we need.
\end{remark}

\subsection{Answering Queries\label{sec:answering}}

In this section, we explain how to use our space efficient ranking
tree in order to guide a range selection query in $T$.

We are given a query $[x_1,x_2]$ as well as a rank $k$, and our goal
is to return the $k$-th smallest $y$-coordinate in the query range.
We begin our search at the root node $v$ of the tree $T$.  In order to
guide the search to the correct child of $v$, we determine the
canonical set of nodes in $R(v)$ that represent the query $[x_1,x_2]$.
Before we query $R(v)$, we search for $x_1$ and $x_2$ in $S_x$.  Let
$r_1$ and $r_2$ denote the ranks of the successor of $x_1$ and
predecessor of $x_2$ in $S$, respectively. We query $R(v)$ using
$[r_1, r_2]$, and use the searchable partial sum data structures
stored in each node of $R(v)$, to identify the canonical set of nodes
in $R(v)$ that represent the interval $[r_1,r_2]$.  At this point we
outline how to use the matrix structures in order to decide how to
branch in $T$.

\paragraph{Matrix Structures:}
We discuss a straightforward, \emph{slow method} of computing the
branch of the child of $v$ to follow.  We then discuss the details of
a faster method, which can also be found in the original
paper~\cite{BGJS10}.

In order to determine the child of $v$ that contains the $k$-th
smallest $y$-coordinate in the query range, recall that $T$ is sorted
by $y$-coordinate.  Let $f_1$ denote the degree of $v$, and $q'_i$
denote the number of points that are contained in the range $[x_1,
  x_2]$ in the subtree rooted at the $i$-th child of $v$, for $1 \le i
\le d$. Determining the child that contains the $k$-th smallest
$y$-coordinate in $[x_1,x_2]$ is equivalent to computing the value
$\tau$ such that $\sum_{i=1}^{\tau-1} q'_i < k$ and $\sum_{i=1}^{\tau}
q'_i \ge k$.  In order to compute $\tau$, we use the matrix structures
in each internal node of the canonical set of nodes, $C$, that
represent $[x_1,x_2]$.  The set $C$ contains $O(\lg n / \lg \lg n)$
internal nodes, as well as at most two leaf nodes.

Consider any internal node $u \in C$, and without loss of generality,
suppose $u$ was on the search path for $r_1$, but not the search path
for $r_2$, and that $u$ has degree $f_2$.  If the search path for
$r_1$ goes through child $c_q$ in $u$, then consider the difference
between columns $f_2$ and $q$ in the first representation of matrix
$M^u$.  We denote this difference as $M'^u$, where $M'^u_i =
M^u_{i,f_2} - M^{u}_{i,q}$, for $1 \le i \le f_1$.  For each internal
node $u \in C$ we add each $M'^u$ to a running total, and denote the
overall sum as $M'$.  Next, for each of the --- at most --- two leaves
on the search path, we query the superblocks of $Y_\ell$ to get the
relevant portions of the sums, and add them to $M'$.  At this point,
$M'_i = q'_i$, and it is a simple matter to scan each entry in $M'$ to
determine the value of $\tau$.  Since each matrix structure has $f_1$
entries in its columns, this overall process takes $O(f_1 \times \lg n
/ \lg \lg n) = O(\lg^{1+\varepsilon_1}n / \lg \lg n)$ time, since
there are $O(\lg n / \lg \lg n)$ levels in $R(v)$.  Since there are
$O(\lg n / \lg \lg n)$ levels in $T$, this costs $O((\lg n / \lg \lg
n)^2 \times \lg^{\varepsilon_1}n)$ time.  This time bound can be
reduced by a factor of $f_1 = O(\lg^{\varepsilon_1}n)$, using
a slightly faster method which we now describe.

\paragraph{Faster method:} The main idea of the fast method is to use
word-level parallelism and the second representation of each matrix in
order to speed up the query time.  When we begin our search in the
root node $v$ of the tree $T$, consider the sections of $k$, denoted
$k_1, ..., k_{g}$.  We query $R(v)$, using the first section $k_1$ to
guide the search.  In order to remove the $f_1$ factor from the slow
method, in each internal node $u$ we subtract the packed words
$w^u_{1,f_2} - w^u_{1,q}$, then add them to a running total $w_1$.
After we have summed the differences between the packed words in all
the internal nodes, we add the relevant sections from the canonical
leaf nodes using Lemma~\ref{lem:leaf-query}.

Since $w_1$ is only a rough approximation of the first section of
$M'$, each value in the packed word might be off by $O(\lg n / \lg \lg
n)$: the number of additions and subtractions we used to compute
$w_1$.  This means there may be errors, caused by carry bits, in
possibly the $O(\lg \lg n)$ least significant bits in each value
stored in the packed word $w_1$.  We scan each entry in $w_1$ to
determine the indices of the first and last entries in $w_1$ that
match $k_1$, except for the last $O(\lg \lg n)$ bits~\footnote{This
  can be done in constant time using parallel subtraction as
  in~\cite{BGJS10}, but that is only necessary in the static case,
  where we are not allowed to spend an additive $O(\lg^{\varepsilon_1}
  n)$ factor at each level in $T$.}, as well as the first value that
is greater than $k_1$ in a more significant bit beyond the $O(\lg \lg
n)$ least significant bits.  Let $K = \{ e_1, ..., e_{b} \}$ denote
this set of entries in the packed word $w_1$.

Next, we check the largest and smallest entries, $e_1$ and $e_{b}$, in
$K$ in order to determine if one of these is $\tau$.  This can be done
in $O(\lg n/ \lg \lg n)$ time. If neither $e_1$ or $e_b$ is $\tau$,
then there are several cases for how to proceed.  If there are only a
constant number of values in $K$, then we can compute the index,
$\tau$, of the child of $v$ that we should branch to, by computing the
entries in $M'$ for these values in $O(\lg n/ \lg \lg n)$ time each.
We call this the \emph{good case}.  However, if there are a
non-constant number of entries in $K$, then we are in the \emph{bad
  case}, and we must do a binary search over $M'$ to determine $\tau$.
This costs $O(\lg \lg n \times \lg n / \lg \lg n) = O(\lg n)$ time in
total.

The key observation, is that after we do the binary search for $\tau$
in the bad case, at \emph{no point in the future} will we ever have to
examine the first section of $k$ or the matrix structures.  This is
because when we are in the bad case, the difference between the first
sections of $q'_{\tau+1}$ and $q'_{\tau}$ is a value that can be
stored in $\Theta(\lg \lg n)$ bits: which is why the overlap between
sections was set to be $\Theta(\lg \lg n)$.  Moreover, since there are
only $g = O(\lg^{\varepsilon_1}n)$ sections, we can spend at most
$O((\lg n)/g \times \lg n) = o((\lg n / \lg \lg n)^2)$ time in bad
cases before we have exhausted all of the bits in the matrix
structures; once there are no more bits, we are guaranteed to have
found the $k$-th smallest $y$-coordinate in the query range.  Since
the good case requires $O(\lg n / \lg \lg n)$ time, and there are
$O(\lg n/ \lg \lg n)$ levels in $T$, our search costs at most $O((\lg
n / \lg \lg n)^2)$ time.

\paragraph{Recursively Searching in $T$:} 
Let $v_\tau$ denote the $\tau$-th child of $v$.  The final detail to
discuss is how we translate the ranks $[r_1,r_2]$ into ranks in the
tree $R(v_\tau)$.  To do this, we query the string $Y(v)$ before
recursing to $v_\tau$.  We use two cases to describe how to find
$Y(v)$ within $Y_\ell$. In the first case, if $v$ is the root of $T$,
then $Y_\ell = Y(v)$.  Otherwise, suppose the range in $Y_{\ell - 1}$
that stores the parent $v_p$ of node $v$ begins at position $z$, and
$v$ is the $i$-th child of $v_p$.  Let $c_{j} = \rangecount(Y(v), z,
z+|Y(v)|,1,j)$ for $1 \le j \le f_1$. Then, the range in $Y_\ell$ that
stores $Y(v)$ is $[z + c_{i-1}, z + c_{i}]$. We then query $Y(v)$, and
set $r_1 = \rankop_\tau(Y(v),r_1)$, $r_2 = \rankop_\tau(Y(v),r_2)$, $k
= k - q'_{\tau-1}$, and recurse to $v_\tau$.  We present the following
lemma, summarizing the arguments presented thus far:

\begin{lemma}
\label{lem:query}
The data structures described in this section allow us to answer range
selection queries in $O((\lg n/ \lg \lg n)^2)$ time.
\end{lemma}

\subsection{Handling Updates}

In this section, we describe the algorithm for updating the data
structures.  We start by describing how insertions are performed.
First, we insert $p$ into $S_x$ and look up the rank, $r_x$, of $p$'s
$x$-coordinate in $S_x$.  Next, we use the values stored in each
internal node in $T$ to find $p$'s predecessor by $y$-coordinate,
$p'$.  We update the path from $p'$ to the root of $T$.  If a node $v$
on this path splits, we must rebuild the ranking tree in the parent
node $v_p$ at level $\ell$, as well as the dynamic string $Y_{\ell}$.

Next, we update $T$ in a top-down manner; starting from the root of
$T$ and following the path to the leaf storing $p$.  Suppose that at
some arbitrary node $v$ in this path, the path branches to the $j$-th
child of $v$, which we denote $v_j$.  We insert the symbol $j$ into
its appropriate position in $Y_\ell$. After updating $Y_\ell$ --- its
leaves in particular --- we insert the symbol $j$ into the ranking
tree $R(v)$, at position $r_x$, where $r_x$ is the rank of the
$x$-coordinate of $p$ among the points in $T(v)$.  As in $T$, each
time a node splits in $R(v)$, we must rebuild the data structures in
the parent node.  We then update the nodes along the update path in
$R(v)$ in a top-down manner: each update in $R(v)$ must be processed
by all of the auxiliary data structures in each node along the update
path.  Thus, in each internal node, we must update the searchable
partial sums data structures, as well as the matrix structures.

After updating the structures at level $\ell$, we use $Y_\ell$ to map
$r_x$ to its appropriate rank by $x$-coordinate in $T(v_j)$.  At this
point, we can recurse to $v_j$.  In the case of deletions, we follow
the convention of Brodal et al.~\cite{BGJS10} and use node marking,
and rebuild the entire data structure after $\Theta(n)$ updates.  We
present the following theorem:

\begin{theorem}\label{thm:main}
Given a set $S$ of points in the plane, there is a linear space
dynamic data structure representing $S$ that supports range selection
queries for any range $[x_1,x_2]$ in $O((\lg n / \lg \lg n)^2)$ time,
and supports insertions and deletions in $O((\lg n / \lg \lg n)^2)$
amortized time.
\end{theorem}

\begin{proof}
The query time follows from Lemma~\ref{lem:query} and the space from
Lemmas~\ref{lem:space-1} and~\ref{lem:space-2}. All that remains is to
analyze the update time.

We can insert a point $p$ into $S_x$ in $O(\lg n)$ time.  The node
structure of $T$ can be updated in $O(\lg^{\varepsilon_1}n)$ amortized
time by the properties of weight balanced B-trees.  Similarly, the
node structure of $R(v)$ for each node $v$ in the update path in $T$
can be updated in $O(\lg^{\varepsilon_2}n)$ amortized time, and there
are $O(\lg n / \lg \lg n)$ ranking trees that are updated.  Thus, the
tree structure of $T$ and the ranking trees in each node can be
updated in $o((\lg n / \lg \lg n)^2)$ time.

The difficulty arises when a node $v$ in $T$ splits, since the index
of $v$ relative to the other children of $v_p$ has changed.  In this
case, we are required to not only rebuild $R(v_p)$, but also the
substring of $Y_\ell$ that stores $Y(v_p)$.  If $T(v)$ contains $m$
points, then $Y(v_p)$ is a string of length $O(m \times
\lg^{\varepsilon_1}n)$, and constructing $R(v_p)$ takes
$O(\lg^{\varepsilon_1}n \times m)$ time after $O(m)$ updates, by
Lemma~\ref{lem:space-2}.  This is $O(\lg^{\varepsilon_1}n \times \lg n
/ \lg \lg n)$ amortized time in total, since splits can occur in each
level of $T$.

One issue is that this analysis assumes that we have access to the
updated version of $Y(v_p)$, storing the indices of the children of
$v_p$, sorted by $x$-coordinate, \emph{after the split}.  We now
explain the technical details of how to compute this string; a task
that requires a few more definitions. Let $v_1$ and $v_2$ denote the
two nodes into which $v$ splits, and let $c_1, ..., c_f$, where $f \in
\Theta(\lg^{\varepsilon_1}n)$ is the degree of $v$, denote the
left-to-right sequence of children of $v$ before the split.  Suppose
that after the split, children $c_1, ..., c_d$ become the children of
$v_1$ and $c_{d+1}, ..., c_f$ become the children of $v_2$, where $1
\le d \le f$, and $d = \Theta(\lg^{\varepsilon_1}n)$.  Also suppose
that $v$ is the $i$-th child of $v_p$, and denote the degree of $v_p$
as $f_p$.

First, we extract the strings $Y(v_p)$ and $Y(v)$ from $Y_\ell$ and
$Y_{\ell+1}$: this requires $O(m \lg^{\varepsilon_1}n + \lg n / \lg
\lg n)$ time in total, since we must traverse a root-to-leaf path in
the dynamic strings.  The next step is to scan both strings $Y(v_p)$
and $Y(v)$ together, and at the same time write an updated string
$Y'(v_p)$, which will be the sequence of indices of $v_p$'s children,
after $v$ is split.  When we encounter the index $c$ in $Y(v_p)$, we
append $c$ to $Y'(v_p)$ if $c \in \{1 ,..., i-1\}$, and $c+1$ if $c
\in \{i+1,...,f_p\}$.  In the case when $c = i$, then we check the
corresponding index $c'$ in $Y(v)$: if $c' \in \{ 1,..., d\}$, then we
append $c$ to $Y'(v_p)$, and a $c+1$ otherwise.  For example, let
$Y(v) = \{ 1,2,3,4,1 \}$, $Y(v_p) = \{1,3,3,2,4,1,2,3,3,4,3,1\}$, and
$v$ be the $3$-rd child of $v_p$.  Suppose that $v$ is to be split so
that children $1$ and $2$ become the children of $v_1$ and $3$ and $4$
become the children of $v_2$.  Then, following the steps we just
described, $Y'(v_p) = \{ 1,3,3,2,5,1,2,4,4,5,3,1\}$.  Overall,
generating $Y'(v_p)$ takes $O(m \lg^{\varepsilon_1}n + \lg n / \lg \lg
n)$ time, since we do one scan through both $Y(v_p)$ and $Y(v)$.  The
additive $O(\lg n / \lg \lg n)$ term is absorbed in all but a constant
number of levels near the bottom of $T$, where $m =
o(\lg^{1-\varepsilon_1}n / \lg \lg n)$.  Thus, the string generation
algorithm described above requires $O(\lg^{\varepsilon_1}n \times \lg
n / \lg \lg n)$ amortized time, when we consider that each level in
$T$ can split during an insertion.

When a split occurs in $T$, we must also do a batched update on the
dynamic string $Y_\ell$, where $\ell$ is the level of $v_p$.  To do
this, we make use of the batched insertion operation from
Lemma~\ref{lem:smallalphabet}.  When $|T(v_p)| > 5W / \Theta(\lg \lg
n) $, where $W$ is the value in Lemma~\ref{lem:smallalphabet}, we can
replace the $O(m \lg^{\varepsilon_1}n)$ values representing $Y(v_p)$
with $Y'(v_p)$ in $O(m \lg^{\varepsilon_1}n)$ time.  However, in the
alternative case, when $|T(v_p)|$ is small, we just directly insert
and delete values into $Y_\ell$ in $O(\lg n / \lg \lg n)$ time per
operation.  As with the string generation algorithm, the case where
$|T(v_p)|$ is too small for batched updates only occurs in a constant
number of bottom levels of $T$.  Thus, the overall cost for updating
$Y_\ell$ for every level in $T$ in which a node is split takes
$O(\lg^{\varepsilon_1}n \times \lg n / \lg \lg n)$ amortized time.

Finally, we consider the more common case where $v$ does not split
during the insertion, and how to update $R(v)$.  Consider the update
path in $R(v)$, and an arbitrary node $u$ on this path.  We can update
the searchable partial sums structure in $u$ in $O(1)$ time in the
worst case by Lemma~\ref{lem:RRR}.  If $u$ is split, then the cost of
rebuilding the searchable partial sums structure is absorbed by the
cost of rebuilding $u$.  The matrix structures can be rebuilt in
$O(1)$ amortized time per internal node on the update path by
Lemma~\ref{lem:matrix-struct}, or $O(\lg n / \lg \lg n)$ amortized
time per ranking tree.  Each conceptual leaf in $R(v)$ takes $O(\lg n/
\lg \lg n)$ time to update by Lemma~\ref{lem:smallalphabet} --- since
we must update the dynamic string --- but there are at most two leaves
updated per ranking tree.  Overall we get that the cost of each update
is $O((\lg n / \lg \lg n)^2)$ amortized time.  
\end{proof}

\section{Dynamic Arrays\label{sec:dynarray}}

In this section, we show how to adapt Theorem~\ref{thm:main} for
problem of maintaining a dynamic array $A$ of values drawn from a
bounded universe $[1 ..\sigma]$.  A query consists of a range in the
array, $[i..j]$, along with an integer $k > 0$, and the output is the
$k$-th smallest value in the subarray $A[i..j]$.  Inserting a value
into position $i$ shifts the position of the values in positions
$A[i..n]$ to $A[i+1..n+1]$, and deletions are analogous.  We present
the following theorem:

\begin{theorem}
\label{thm:dynarray}
Given an array $A[1..n]$ of values drawn from a bounded universe
$[1..\sigma]$, there is an $n H_0(A) + o(n \lg \sigma) + O(w)$ bit
data structure that can support range selection queries on $A$ in
$O(\frac{\lg n}{\lg \lg n} \frac{\lg \sigma}{\lg \lg \sigma})$ time,
and insertions into, and deletions from, $A$ in $O(\frac{\lg n}{\lg
  \lg n} \frac{\lg \sigma}{\lg \lg \sigma})$ amortized time.
\end{theorem}

\begin{proof}
  The data structure is roughly the same as the tree $T$ from
  Theorem~\ref{thm:main}, except that now we need a few extra
  techniques to avoid paying for more than a constant number of
  pointers.  The first idea is to use a generalized wavelet tree $T'$
  with fan out $O(\lg^{\varepsilon_1} \sigma)$ over the universe $[1
  ..\sigma]$, instead of the original weight balanced B-tree $T$,
  (c.f., Ferragina et al.~\cite{FMMN04}).  The tree $T'$ has height
  $\Theta(\lg \sigma / \lg \lg \sigma)$, and, as in the tree $T$, we
  store dynamic strings $Y_\ell$ at each level in $T'$, as well as
  ranking trees for each node of $T'$.  The reason we cannot afford to
  increase the branching factor beyond $\Theta(\lg^{\varepsilon_1}
  \sigma)$ is that the cost of a search in the ranking tree during a
  \emph{bad case} (see the description of the \emph{faster method}
  from Section~\ref{sec:answering}) would dominate the query cost: if
  the fan out of $T'$ is $f'$, then the total time in bad cases can be
  $O((f' \lg f' \lg n) / \lg \lg n)$.

  Another issue is that the pointers to the ranking trees stored in
  each node of $T'$ occupy too much space: the lowest level in $T'$ in
  which we store ranking trees has $O(n (\lg \lg n)^2/ (\lg^{2}n
  \lg^{\varepsilon_1}\sigma))$ nodes, and therefore $O(nw (\lg \lg
  n)^2/( \lg^{2}n \lg^{\varepsilon_1}\sigma))$ bits are required for
  pointers to the roots of the ranking trees.  However, since the
  maximum number of bits occupied by ranking trees at any level is at
  most $O\left( n(\lg \lg n)^2 \lg^{\varepsilon_1} \sigma / \lg n
  \right)$, which is $o(n)$ --- not counting the pointers to their
  roots --- we can use the same technique as in
  Lemma~\ref{lem:space-2} and group all the ranking trees at a given
  level in $T'$ into a fixed memory area of size $o(n)$.  Thus, we can
  replace all of the $O(w)$ bit pointers to the roots of the ranking
  trees with $O(\lg n)$ bit pointers, and the space for pointers to
  the ranking trees becomes $o(n) + O(w \lg \sigma / \lg \lg \sigma)$:
  one pointer to the fixed memory area per level in the tree.  We also
  need $O(w \lg \sigma / \lg \lg \sigma)$ bits to store the pointers
  to the dynamic strings at each level.  At this point, we can use the
  technique of M{\"a}kinen and Navarro~\cite{MN06} to further reduce
  this to $O(w)$ bits.

  Finally, since we are working solely with ranks rather than
  $x$-coordinates, we can discard the red-black tree $S_x$.  To
  analyze the total space cost we sum the space required by the
  dynamic strings and the ranking trees at each level in $T$.  Thus,
  the total number of bits required is:
$$\sum_{\ell = 1}^{h_1} \left(n H_0(Y_\ell) + O\left( \frac{n \lg \lg \sigma (\lg \lg n)^2}{\sqrt{\lg n}} + \frac{n (\lg \lg n)^2 \lg^{\varepsilon_1}\sigma}{\lg n}\right) \right) + O(w)$$
\noindent
By the same arguments presented in~\cite{FMMN04}, this simplifies to:
$$n H_0(A) + O\left(\frac{n \lg \sigma (\lg \lg n)^2}{\sqrt{\lg n}}\right) + O(w)$$
\end{proof}

\section{Concluding Remarks}
In the same manner as Brodal et al.~\cite{BGJS10}, the data structure
we presented can also support orthogonal range counting queries in the
same time bound as range selection queries.  We note that the
cell-probe lower bounds for the static range median and static
orthogonal range counting match~\cite{P07,JL11}, and --- very recently
--- dynamic weighted orthogonal range counting was shown to have a
cell-probe lower bound of $\Omega((\lg n / \lg \lg n)^2)$ query time
for any data structure with polylogarithmic update time~\cite{L11}.
In light of these bounds, it is likely that $O((\lg n / \lg \lg n)^2)$
time for range median queries is optimal for linear space data
structures with polylogarithmic update time.  However, it may be
possible to do better in the case of dynamic range selection, when $k
= o(n^\varepsilon)$, for any $\varepsilon > 0$, using an adaptive data
structure as in the static case~\cite{JL11}.

\paragraph{Acknowledgements:} We would like to thank Gelin Zhou for
pointing out an error in the preliminary version, in the proof of
Theorem 2.

\bibliographystyle{splncs03} \bibliography{median}

\end{document}